\documentclass[submission,copyright,creativecommons]{eptcs}
\providecommand{\event}{FROM 2023} % Name of the event you are submitting to

\usepackage{tikz}
\usepackage{amsmath,amssymb}
\usepackage{amsthm, cleveref}
\newtheorem{theorem}{Theorem}
\newtheorem{lemma}{Lemma}
\newtheorem{definition}{Definition}
\newtheorem{proposition}{Proposition}
\newtheorem{example}{Example}
\newtheorem{corollary}{Corollary}

\def\tT{\mathbb{T}}
\def\bN{\mathbb{N}}

\def\Rgd{\mathit{Rgd}}
\def\leaf{\mathit{Leaf}}

\def\cliques{\mathit{cliques}}
\def\cvr{\mathit{cfg}}

\def\cI{\mathcal{I}}
\def\cJ{\mathcal{J}}
\def\oC{\overline{C}}

\def\sse{\subseteq}
\def\es{\emptyset}
\def\cC{\mathcal{P}}

\def\repr{\mathit{repr}}
\def\cP{\mathcal{P}}

\usepackage{iftex}

\ifpdf
  \usepackage{underscore}         % Only needed if you use pdflatex.
  \usepackage[T1]{fontenc}        % Recommended with pdflatex
\else
  \usepackage{breakurl}           % Not needed if you use pdflatex only.
\fi

\title{Enumerating All Maximal Clique-Partitions of an Undirected Graph}
\author{Mircea Marin
\institute{West University of Timi\c soara\\ Timi\c soara, Romania}
\email{mircea.marin@e-uvt.ro}
\and
Temur Kutsia \qquad\qquad Cleo Pau
\institute{Research Institute for Symbolic Computation\\
Johannes Kepler University\\
Linz, Austria}
\email{\quad kutsia@risc.jku.at \quad\qquad ioana.pau@risc.jku.at}
\and Mikheil Rukhaia
\institute{Institute of Applied Mathematics\\ Tbilisi State University\\ Tbilisi, Georgia}
\email{mrukhaia@gmail.com}
}
\def\titlerunning{Enumerating All Maximal Clique-Partitions}
\def\authorrunning{M. Marin, T. Kutsia, C. Pau \& M. Rukhaia}
\begin{document}
\maketitle

\begin{abstract}
We address the problem of enumerating all maximal clique-partitions of an undirected graph and present an algorithm based on the observation that every maximal clique-partition can be produced from the maximal clique-cover of the graph by assigning the vertices shared among maximal cliques, to belong to only one clique. 
This simple algorithm has the following drawbacks: (1) the search space is very large; (2) it finds some clique-partitions which are not maximal; and (3) some  clique-partitions are found more than once. 
We propose two criteria to avoid these drawbacks. The outcome is an algorithm that explores a much smaller search space and guarantees that every maximal clique-partition is computed only once. 

The algorithm can be used in problems such as anti-unification with proximity relations or in resource allocation tasks when one looks for several alternative ways to allocate resources.
\end{abstract}
{\bf Acknowledgments.} Partially supported by the Austrian Science Fund (FWF) under project P 35530 and by the Shota Rustaveli National Science Foundation of Georgia under project FR-21-16725.
\section{Introduction}
\label{sec:intro} 
 In this paper, we are interested in computing all maximal clique-partitions in a graph. The original motivation comes from anti-unification with proximity relations. Anti-unification is a well-known technique in computational logic. It was introduced in \cite{Plotkin70,Reynolds70} and was quite intensively investigated in the last years, see, e.g. \cite{DBLP:journals/fss/Ait-KaciP20,DBLP:journals/amai/AlpuenteEMS22,DBLP:conf/cade/KutsiaP22,DBLP:conf/tbillc/KutsiaP19,https://doi.org/10.48550/arxiv.2302.00277}. Given two first-order logic terms $t_1$ and $t_2$, it aims at computing a least general generalization of those terms. That means, one is looking for a term $s$ from which $t_1$ and $t_2$ can be obtained by variable substitutions. Such an $s$ is called a generalization of $t_1$ and $t_2$. Moreover, there should be no other generalization $r$ of $t_1$ and $t_2$, which can be obtained from $s$ by a substitution. For instance, if $t_1$ and $t_2$ are the ground terms  $f(a,a)$ and $f(b,b)$, then anti-unification computes their least general generalization $f(x,x)$. Replacing variable $x$ by $a$ (resp. by $b$) in it, one gets $f(a,a)$ (resp. $f(b,b)$). Note that $f(x,y)$ and  $x$ are also generalizations of $f(a,a)$ and $f(b,b)$, but they are not least general. Anti-unification has been successfully used in inductive reasoning, inductive logic programming, reasoning and programming by analogy, term set compression, software code clone detection, etc.
 
 In many applications, that can be also relevant for anti-unification, one has to deal with imprecise or vague information. In such circumstances, one tends to consider two objects the same, if they are ``sufficiently close'' to each other. However, such a proximity relation is not transitive. %: For instance, if one considers two cities being close to each other if the distance between them is not more than 200 km, then Salzburg is close to Linz (133 km) and Linz is close to Vienna (185 km), but Salzburg is not close to Vienna (300 km). 
 Nontransitivity has to be dealt with in a special way. Proximity relations (reflexive symmetric fuzzy binary relations) characterize the notion of `being close' numerically. They become crisp once we fix the threshold from which on, the distance between the objects can be called `close'. 

 Symbolic constraint solving (for unification, matching, and anti-unification constraints) over proximity relations has been studied recently by various authors, e.g., \cite{DBLP:conf/cade/KutsiaP22,PauThesis,DBLP:journals/fss/Ait-KaciP20,DBLP:journals/fss/IranzoR15,DBLP:journals/tfs/IranzoS21}. The approaches can be characterized as class-based and block-based. Considering proximity relations as (weighted) undirected graphs, a proximity class of a vertex is its neighborhood (i.e., the set of vertices to which the current vertex is connected by an edge), while a proximity block is a clique. In the class-based approach to proximity constraint solving, two objects are considered proximal if one of them belongs to the proximity class of another. In the block-based approach, two objects are proximal if they belong to the same \emph{unique} maximal proximity block.  
The block-based approach is one that is closely related to the subject of this paper. To compute a minimal complete set of generalizations of two first-order logic terms with this approach, one needs to consider all maximal clique-partitions of the graph induced by the proximity relation between constants and between function symbols. For instance, if $a$ is close to both $b$ and $c$, but $b$ and $c$ are not close to each other, then $f(a,a)$ and $f(b,c)$ have two minimal common generalizations: $f(a,x)$ and $f(x,a)$. In this example, the proximity graph would be $(\{ a,b,c\}, \{(a,b),(a,c)\})$. It has two maximal clique partitions $\{\{a,b\}, \{c\}\}$ and $\{ \{a,c\},\{b\} \}$ that tell exactly which symbols should be considered the same. In the first case these are $a$ and $b$, leading to the generalization $f(a,x)$, and in the second case they are $a$ and $c$, giving $f(x,a)$. 
Also, in the block-based approach to approximate unification, one would need to maintain maximal clique-partitions of the proximity graph in order to detect that, e.g., $f(x,x)$ and $f(b,c)$ are not unifiable in the abovementioned proximity relation, see, e.g.,~\cite{DBLP:journals/tfs/IranzoS21}.
 
 %In a lighter view, all maximal clique-partitions can be useful in designing a menu of various possible courses cooked with (leftover) ingredients which are not enough to be used in more than one dish. Parents of kids who are reluctant to eat would appreciate such a choice. 
 
Also, the resource allocation problem, when one looks for several alternative ways to allocate resources, can be an application area of the algorithm considered in this paper.

Whereas the problem of computing all maximal cliques is well studied \cite{DBLP:journals/cacm/BronK73,DBLP:journals/tcs/TomitaTT06,DBLP:conf/walcom/Tomita17,DBLP:journals/tcs/CazalsK08}, the problem of computing all maximal clique-partitions became of interest only recently. To the best our knowledge, the only previous study of it is the one reported in 2022 by C. Pau her PhD thesis \cite[Sect. 3.3.2]{PauThesis}. In this paper we provide a more in-depth analysis  of the problem and propose another algorithm which performs better than the one described in \cite{PauThesis}.
For a given undirected graph $G$, in order to compute all its maximal clique-partitions, we use a kind of top-down approach. First, we compute the maximal clique cover of $G$, and the list $S$ of all graph vertices which are shared among maximal cliques. 
%and generate clique-partitions of $G$ by deciding for every vertex in $S$ to belong to only one clique. 
By a systematic enumeration of all possibilities to assign each vertex in $S$ to only one clique where it belongs, we obtain an algorithm that finds all clique-partitions of $G$ in a tree-like search space starting from the maximal clique-cover of $G$. Our algorithm is optimal in the following sense: (1) It computes only maximal clique-partitions, by avoiding computations below some nodes of the search space which yield non-maximal clique-partitions, and (2) Each maximal clique-partition is computed only once. Moreover, our algorithm computes the maximal clique-partitions incrementally, in the following sense:  If one does not want to get all  solutions, he/she can stop the algorithm after computing a certain number of solutions. As a result, the computation of maximal clique-partitions can be streamlined with other operations on them.

The  paper is structured as follows.  In Section \ref{sect:prelim} we introduce some preliminary notions, the search space $\tT^S(G)$ for maximal clique-partitions of an undirected graph $G$ and its main properties. The following two sections describe our main contributions: a criterion to avoid the computation of nonmaximal clique-partitions (Section \ref{sect:nonmax}), and a criterion to avoid the redundant computations of the same maximal clique-partition (Section  \ref{sect:norepeat}).
In Section \ref{sect:combo} we indicate how to combine these two criteria and define an algorithm to enumerate all maximal clique-partitions of an undirected graph. An analysis of the runtime complexity of our algorithm is performed in Section \ref{sect:complexity}. In the last section we draw some conclusions.

\section{Preliminaries}
\label{sect:prelim}
We consider undirected graphs $G=(V,E)$ where $V$ is the set of nodes and $E$ is the set of edges. A {\bf clique} in $G$ is a nonempty subset of $V$ such that every two vertices of it are incident. A clique $C$ is {\bf maximal} if it is not a proper subset of another clique.

A {\bf cover} of $G$ is a finite family $\{V_1,\ldots,V_m\}$ of nonempty sets of nodes such that $\bigcup_{i=1}^mV_i=V$.
A {\bf clique-cover} of $G$ is a cover $\cC$ of $V$ such that every $C\in\cC$ is a clique in $G$.
A partition into cliques, or shortly {\bf clique-partition} of $G$, is a partition $\cC$ of $V$ such that every $C\in\cC$ is a clique in $G$. $\cC$ is a {\bf maximal clique-partition} of $G$ if $\cC$ is a clique-partition of $G$ and $\cC$ does not contain two different cliques $C,C'$ such that $C\cup C'$ is a clique.

 A graph may have several maximal clique-partitions. In the literature, a problem that was studied intensively is to compute a maximal clique-partition with the smallest number of cliques. Tseng's algorithm~\cite{tseng}, introduced to solve this problem, was motivated by its application in the design of processors. Later, Bhasker and Samad~\cite{bhasker} proposed two other algorithms. They also derived the upper bound on the number of cliques in a partition and showed that there exists a partition containing a maximal clique of the graph. 
 A problem closely related to clique-partition is the vertex coloring problem, which requires to color the vertices of a graph in such a way that two adjacent vertices have different colors. In fact, a clique-partitioning problem of a graph is equivalent to the coloring problem of its complement graph. Both problems are NP-complete~\cite{DBLP:conf/coco/Karp72}.

We write $\bN$ for the set of natural numbers starting from 1, and $\bN^*$ for the monoid of finite sequences of numbers from $\bN$ with the operation of sequence concatenation and neutral element $\epsilon$. If $n\in\bN$ we assume that $[n]$ is the set of natural numbers $k$ such that $1\leq k\leq n$. 

From now on we assume that $G=(V,E)$ is an undirected graph for which we know:
\begin{enumerate}
\item An enumeration $\cvr_0:=[\oC_1,\ldots,\oC_m]$ of all maximal cliques of $G$. We denote the maximal cliques of $G$ with identifiers with an overbar.
The value $m$ indicates the number of maximal cliques of graph $G$.
\item For every vertex $v\in V$ and set of nodes $C\sse V$ we define: 
\begin{itemize}
\item $\cliques(v):=\{i\mid v\in \oC_i\}$, and  $d(v):=|\cliques(v)|$,
\item $\cliques(C):=\bigcap_{v\in C}\cliques(v)$ for every set of nodes $C\sse V$. 
\end{itemize}
\item An enumeration $S:=[v_1,\ldots,v_s]$ of all nodes $v\in V$ with $d(v)>1$.
\item $\Rgd:=\{k\in[m]\mid \text{there is a vertex }v\in V\text{ with }\cliques(v)=\{k\}\}.$
\end{enumerate}

$\cliques(v)$ is the set of indices of maximal cliques where node $v$ belongs, and $\cliques(C)$ is the set of indices of maximal cliques which contain $C$. Thus, a nonempty set of nodes $C$ is not a clique iff $\cliques(C)=\es$.
The nodes $v\in S$ are those that belong to more than one maximal clique. To transform the maximal clique cover $\{\oC_1,\ldots,\oC_m\}$ into a clique partition, we must assign every node $v\in S$ to belong to only one clique. To formalize this process, we introduce a couple of auxiliary notions. 

A {\bf configuration} is an enumeration 
$[C_1,\ldots,C_m]$ of empty sets or cliques of $G$, such that $\bigcup_{k=1}^m C_k=V$ and $C_i\sse\oC_i$ for all $1\leq i\leq m$. 
In particular, $\cvr_0=[\oC_1,\ldots,\oC_m]$ is a configuration.
Every configuration $\cvr=[C_1,\ldots,C_m]$ represents a clique cover denoted by $$\repr(\cvr):=\{C_i\mid 1\leq i\leq m\text{ and }C_i\neq \es\}.$$

We distinguish two sets of configurations of interest: the set $\Pi_{\mathit{cp}}(G)$ of configurations $\cvr$ for which $\repr(\cvr)$ is a clique partition of $G$; and the set $\Pi_{\mathit{mcp}}(G)$ is the set of configurations $\cvr$ for which $\repr(\cvr)$ is a maximal clique partition of $G$.
\begin{lemma}
\label{mcp-cp}
Every maximal clique-partition of $G$ is represented by a configuration in $\Pi_{\mathit{mcp}}(G)$.
\end{lemma}
\begin{proof}
Let $\cP$ be a maximal clique-partition of $G$. Then for every $C\in\cP$ there exists a maximal clique $\varphi(C)$ such that $C\sse\varphi(C)$. 
If $C_1,C_2\in\cP$ and $\varphi(C_1)=\varphi(C_2)$ then $C_1\cup C_2\sse\varphi(C_1)$ is a clique, 
and the maximality of $\cP$ implies $C_1=C_2$. Thus $\varphi$ is injective and we can define $\cvr=[C_1,\ldots,C_m]$ by
$$C_k:=\left\{\begin{array}{ll}
C&\text{if }C\in\cP\text{ and }\varphi(C)=\oC_k,\\
\es&\text{otherwise}
\end{array}\right.$$
for all $k\in[m]$. Then $\repr(\cvr)=\cP$ because $\varphi$ is injective.
Thus $\cvr\in\Pi_{\mathit{mcp}}(G)$.
\end{proof}
For every node $v\in V$ we define the relation 
$[C_1,\ldots,C_m]\to_{(v,i)}[C'_1,\ldots,C'_m]$ to hold if
$v\in C_i$ for some $1\leq i\leq m$, $C'_i=C_i$ and $C'_j=C_j-\{v\}$ for all $j\in[m]-\{i\}.$ This relation corresponds to the decision to assign node $v_i$ to the $i$-th clique of the configuration.

\subsection{The search space $\tT^S(G)$}
It is easy to see that, if $S=[v_1,v_2,\ldots,v_s]$ and $i_1\in \cliques(v_1),i_2\in \cliques(v_2)$, \ldots, $i_s\in\cliques(v_s)$ then $\cvr_0=[\oC_0,\ldots, \oC_m] \to_{(v_1,i_1)} \cvr_1 \to_{(v_2,i_2)}\ldots\to_{(v_s,i_s)}\cvr_s$ is a sequence of decision steps that ends with a configuration whose representation is a partition of $G$. 

We let $\tT^S(G)$ be the tree with root $\cvr_0$ and edges $\cvr\to_{(v,i)}\cvr'$ which correspond to the decision to keep the shared node $v\in S$ in the $i$-th component of the configuration.
We will use $\tT^S(G)$ as the search space for maximal clique-partitions, and let $\leaf(\tT^S(G))$ be the set of leaf configurations in $\tT^S(G)$.
\begin{example}
\label{empi}
The simple graph\quad
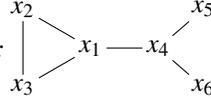
\begin{tikzpicture}[inner sep=1.5pt,scale=.6,baseline=-.2em]
\node[left] at (-.8,0) {\small $G$:};
\node (4) at (2.5,0) {\small $x_4$};
\node (5) at (3.5,.866) {\small $x_5$};
\node (1) at (1,0) {\small $x_1$};
\node (2) at (-0.5,0.866) {\small $x_2$};
\node (6) at (3.5,-.866) {\small $x_6$};
\node (3) at (-.5,-.866) {\small $x_3$};
\draw  (1) -- (2) -- (3) -- (1) -- (4)-- (5)  (4) -- (6); 
\end{tikzpicture}\\[.4em]
has four maximal cliques: $\oC_1=\{x_1,x_2,x_3\}$, $\oC_2=\{x_1,x_4\}$, $\oC_3=\{x_4,x_5\}$, $\oC_3=\{x_4,x_6\}$ and two nodes shared among maximal cliques: $S=[v_1,v_2]$ where $v_1=x_4$, $v_2=x_1$. In this example we have $\cliques(x_1)=\{1,2\}$, $\cliques(x_2)=\{2,3,4\}$. The exhaustive search space for maximal clique partitions is the tree $\tT^S(G)$ depicted below:
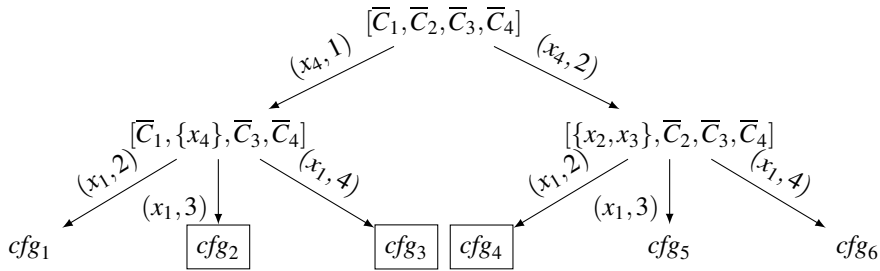
\begin{center}
\begin{tikzpicture}[>=latex]
\node (s0) at (0,1.5) {\small $[\oC_1,\oC_2,\oC_3,\oC_4]$};
\node (s11) at (-3,0) {\small $[\oC_1,\{x_4\},\oC_3,\oC_4]$};
\node (s111) at (-5.5,-1.5) {\small $\cvr_1$};
\node[draw] (s112) at (-3,-1.5) {\small $\cvr_2$};
\node[draw] (s113) at (-.5,-1.5) {\small $\cvr_3$};
\node[draw] (s121) at (0.5,-1.5) {\small $\cvr_4$};
\node (s122) at (3,-1.5) {\small $\cvr_5$};
\node (s123) at (5.5,-1.5) {\small $\cvr_6$};
\node (s12) at (3,0) {\small $[\{x_2,x_3\},\oC_2,\oC_3,\oC_4]$};
\draw[->] (s11) -- node[sloped,above] {\small $(x_1,2)$}(s111);
\draw[->] (s11) -- node[left,near end] {\small $(x_1,3)$}(s112);
\draw[->] (s11) -- node[sloped,above] {\small $(x_1,4)$}(s113);
\draw[->] (s12) -- node[sloped,above] {\small $(x_1,2)$}(s121);
\draw[->] (s12) -- node[left,near end] {\small $(x_1,3)$}(s122);
\draw[->] (s12) -- node[sloped,above] {\small $(x_1,4)$}(s123);
\draw[<-] (s12) -- node[sloped,above]{\small $(x_4,2)$} (s0);
\draw[->] (s0)-- node[sloped,above]{\small $(x_4,1)$} (s11);
\end{tikzpicture}
\end{center}
where the leaf configurations are \\ [.5em]
 $\cvr_1=[\oC_1,\{x_4\},\{x_5\},\{x_6\}]$ with $\repr(\cvr_1)=\{\oC_1,\{x_4\},\{x_5\},\{x_6\}\}$.\\
 $\cvr_2=[\oC_1,\es,\oC_3,\{x_6\}]$ with  $\repr(\cvr_2)=\{\oC_1,\oC_3,\{x_6\}\}$. \\
$\cvr_3=[\oC_1,\es,\{x_5\},\oC_4]$ with $\repr(\cvr_3)=\{\oC_1,\{x_5\},\oC_4\}$. \\
$\cvr_4=[\{x_2,x_3\},\oC_2,\{x_5\},\{x_6\}]$ with $\repr(\cvr_4)=\{\{x_2,x_3\},\oC_2,\{x_5\},\{x_6\}\}$.\\
$\cvr_5=[\{x_2,x_3\},\{x_1\},C_3,\{x_6\}]$  with $\repr(\cvr_5)=\{\{x_2,x_3\},\{x_1\},C_3,\{x_6\}\}$.\\
$\cvr_6=[\{x_2,x_3\},\{x_1\},\{x_5\},C_4]$ with $\repr(\cvr_6)=\{\{x_2,x_3\},\{x_1\},\{x_5\},C_4\}$.

Only the final configurations $\cvr_2,\cvr_3,\cvr_4$ represent maximal clique-partitions of $G$:\\
$\cP_1=\{\oC_1,\oC_3,\{x_6\}\}=\repr(\cvr_2)$,\\ 
$\cP_2=\{\oC_1,\{x_5\},\oC_4\}=\repr(\cvr_3)$, and\\ 
$\cP_3=\{\{x_2,x_3\},\oC_2,\{x_5\},\{x_6\}\}=\repr(\cvr_4)$.\\
The other final configurations in $\tT^S(G)$ represent non-maximal clique-partitions of $G$.\qed
\end{example}
\subsubsection{Properties of the search space $\tT^S(G)$}
The following are immediate consequences of the definition: if $S=[v_1,v_2,\ldots,v_s]$ is the list of nodes shared among the maximal cliques of $G$ then 
\begin{enumerate}
\item $\tT^S(G)$ has depth $s$, and all its leaf configurations occur at depth $s$. 
\item Every internal configuration at depth $\ell<s$ in $\tT^S(G)$ has $d(v_\ell)$ children.
\end{enumerate}
For every configuration $\cvr$ in $\tT^S(G)$ there is a unique path
$$\cvr_0\to_{(v_1,i_1)}\cvr_1\to_{(v_2,i_2)}\ldots\to_{(v_\ell,i_\ell)}\cvr$$ from the root configuration to $\cvr$. We let $\delta(\cvr):=[i_1,\ldots,i_\ell]$ be the sequence of assignment decisions made for the shared nodes $v_1,\ldots,v_\ell\in S$. 

\begin{lemma}
\label{core-lemma}
If $\cvr\in\tT^S(G)$ with $\delta(\cvr)=[i_1,\ldots,i_\ell]$ then all descendants $[C_1,\ldots,C_m]$ of $\cvr$ in $\tT^S(G)$, including $\cvr$, have $C_i\neq\es$ for every $i\in \Rgd\cup\{i_1,\ldots,i_\ell\}$. 
\end{lemma}
\begin{proof}
If $i=i_p\in\{i_1,\ldots,i_l\}$ then the shared node $v_p\in S$ was assigned to the clique with index $i$, thus $C_i\neq \es$ because $v_p
\in C_i$.
If $i\in\Rgd$ then $\oC_i$ has a node $v$ with $\cliques(v)=\{i\}$. Node $v$ persists in the $i$-th component of all configurations in $\tT^S(G)$. In particular, $C_i\neq\es$ because $c\in C_i$. 
\end{proof}
\begin{lemma}
\label{gambit}
$\Pi_{mcp}(G)\sse\leaf(\tT^S(G))$.
\end{lemma}
\begin{proof}
Let $C=[C_1,\ldots,C_m]\in\Pi_{mcp}(G)$.
For every $v\in S$, there is a unique  $\kappa(v)\in [m]$ such that $v\in C_{\kappa(v)}$. Then $\cvr_0\to_{(v_1,\kappa(v_1))}\cvr_1
\to_{(v_2,\kappa(v_2))}\ldots\to_{(v_s,\kappa(v_s))}\cvr_s$ is a valid sequence of decision steps. Moreover, $\cvr_s$ is a final configuration in $\tT^S(G)$ and $\repr(\cvr_s)=\cP$.  
\end{proof}
\begin{corollary}
Every maximal clique-partition is represented by a configuration in $\tT^S(G)$.
\end{corollary}
\begin{proof}
Immediate consequence of Lemmas \ref{mcp-cp} and \ref{gambit}.
\end{proof}
From these preliminary results, we derive the following algorithm to find all clique-partitions of $G$: we traverse systematically (e.g., in a depth-first manner) the search space $\tT^S(G)$, and for every final configuration $\cvr$ in $\tT^S(G)$ we check if $\repr(\cvr)$ is a maximal clique-partition of $G$. This method has the following drawbacks:
\begin{enumerate}
\item The search space can be huge, with many final configurations for non-maximal clique-partitions. For instance, in Example \ref{empi}, the final configurations $\cvr_1$, $\cvr_6$ and $\cvr_6$ represent non-maximal clique partitions.
\item Some maximal clique-partitions may be represented by more than one final configuration. 
\end{enumerate}
We wish to prune the search space as much as possible, to eliminate the computation of configurations for non-maximal clique-partitions, and to ensure the computation of exactly one configuration for every maximal clique-partition. 

\section{Avoiding the computation of nonmaximal clique-partitions}
\label{sect:nonmax}
A final configuration $[C_1,\ldots,C_m]$ does not represent a maximal clique-partition if there exist $1\leq i\neq j\leq m$ such that $C_i\neq\es\neq C_j$ and $C_i\cup C_j$ is a clique. Therefore, it is useful to detect and stop searching below nodes of $\tT^S(G)$ from where we reach only such final configurations.
\begin{definition}
We say that a configuration $\cvr=[C_1,\ldots,C_m]\in\tT^S(G)$ with $\delta(\cvr)=[i_1,\ldots,i_\ell]$ is a {\bf $T1$-node}, or that it has property $T1$, if either
\begin{description}
\item[(a)] $\cliques(C_a)\cap Rgd\neq \es$ for some clique index $a\in\{i_1,\ldots,i_\ell\}-Rgd$, or
\item[(b)] $\cliques(C_a)\cap\cliques(C_b)\neq\es$ for  distinct clique indices $a,b\in \{i_1,\ldots,i_\ell\}$.  
\end{description}
We let $\tT^S_{!(T_1)}(G)$ be the result of pruning from $\tT^S(G)$ all subtrees whose root is a $T1$-node. 
\end{definition}
\begin{proposition}
\label{forbid2}
If $\cvr$ is a $T1$-node of $\tT^S(G)$ then there is no final configuration $[C_1,\ldots,C_m]$ below or equal to $\cvr$ such that $\repr(\cvr)$ is a maximal clique-partition.
\end{proposition}
\begin{proof}
Let $[C_1,\ldots,C_m]$ be a final configuration below or equal to $\cvr$ in $\tT^S(G)$.

(a) If $\cliques(C_a)\cap Rgd\neq \es$ for some clique index $a\in\{i_1,\ldots,i_\ell\}-Rgd$ then $C_a\neq\es$ by Lemma \ref{core-lemma}, and there exists 
$b\in Rgd$ such that $C_a\sse \oC_b$. In this case we have: (1) $b\neq a$ because $b\in \Rgd$ and $a\not\in \Rgd$; (b) $C_b\neq\es$ by Lemma \ref{core-lemma}; and (3) $C_a\cup C_b$ is a clique included in $\oC_b$ because $C_a\sse\oC_b$ and $C_b\sse\oC_b$. Therefore, $\cP$ is not a maximal clique-partition.  

(b) If $\cliques(C_a)\cap\cliques(C_b)\neq\es$ for  distinct  $a,b\in \{i_1,\ldots,i_\ell\}$ then
$C_a\neq\es\neq C_b$ by Lemma \ref{core-lemma}, and there exists $p\in \cliques(C_a)\cap\cliques(C_b)$ such that 
$C_a\sse\oC_p$ and $C_b\sse\oC_p$. Thus $C_a\cup C_b\sse \oC_p$, hence $C_a\cup C_b$ is a clique and $\cP$ is not maximal clique-partition.
\end{proof}
\begin{corollary}
$\leaf(\tT^S_{!(T_1)}(G))=\Pi_{mcp}(G)$.
\end{corollary}
\begin{proof}
Immediate consequence of Lemma \ref{gambit}, Proposition \ref{forbid2} and the obvious observation that every $\cvr\in\tT^S(G)$ with $\repr(\cvr)$ non-maximal is not in $\tT^S_{!(T1)}(G)$ because it is a $T1$-node.
\end{proof}
\begin{example}
\label{ex:dumbo}
An enumeration of the maximal cliques of the undirected graph $G$:
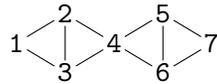
\begin{center}
\begin{tikzpicture}[inner sep=1pt,scale=1.3]
\node (1) at (0,0) {\tt 1};
\node (2) at (0.5,.3) {\tt 2};
\node (3) at (0.5,-.3) {\tt 3};
\node (4) at (1,0) {\tt 4};
\node (5) at (1.5,.3) {\tt 5};
\node (6) at (1.5,-.3) {\tt 6};
\node (7) at (2,0) {\tt 7};
\draw (2) -- (1) -- (3) -- (4) -- (2) -- (3) (5) -- (4) -- (6) -- (7) -- (5) -- (6);
\end{tikzpicture}
\end{center}
is $\cvr_0:=[\oC_1,\oC_2,\oC_3,\oC_4]$ where $\oC_1=\{1,2,3\}$, $\oC_2=\{2,3,4\}$, $\oC_3=\{4,5,6\}$, $\oC_4=\{5,6,7\}$. Then $d[1]=d[7]=1$ and $d[v]=2$ for all vertices $v\in\{2,3,4,5,6\}$. Therefore $\Rgd=\{1,4\}$ and we can choose $S=[2,3,4,5,6]$. The tree $\tT^S(G)$ has $\sum_{i=1}^5\prod_{j=1}^i2=62$ non-root configurations and $2^5=32$ final configurations, whereas $\tT^S_{!(T1)}(G)$ has 25 non-root configurations, as shown in Fig. \ref{fig:dt1}. The final configurations in $\tT_{!(T1)}^S(G)$ are $\cvr_1,\cvr_2,\cvr_3,\cvr_4,\cvr_5,\cvr_7,\cvr_9,\cvr_{11}$, and
$$\begin{array}{rr}
\repr(\cvr_1)=\repr(\cvr_5)=\{\{1,2,3\},\{4,5,6\},\{7\}\},&\repr(\cvr_2)=\{\{1,2,3\},\{4,5,6\},\{7\}\},\\
\repr(\cvr_3)=\{\{1,2,3\},\{4,5\},\{6,7\}\},&\repr(\cvr_4)=\{\{1,2,3\},\{4,6\},\{5,7\}\},\\
\repr(\cvr_7)=\{\{1,2\},\{3,4\},\{5,6,7\}\},&\repr(\cvr_9)=\{\{1,3\},\{2,4\},\{5,6,7\}\},\\
\repr(\cvr_{11})=\{\{1\},\{2,3,4\},\{5,6,7\}\}.
\end{array}
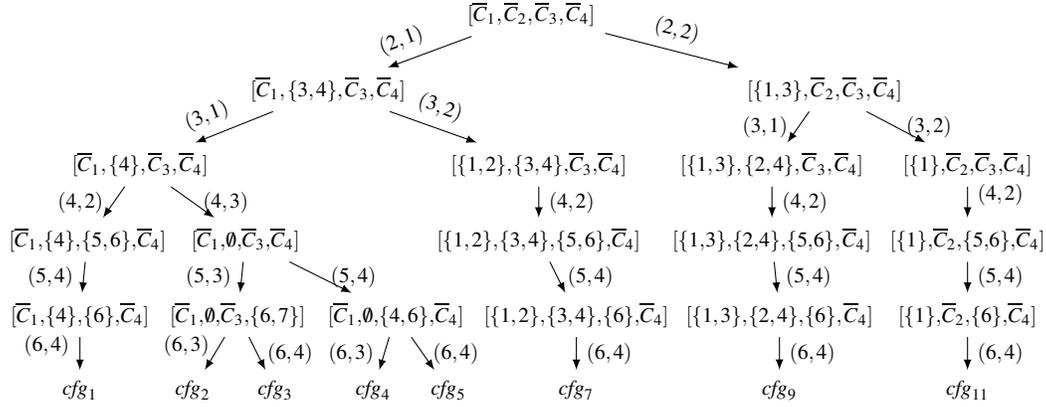
$$

\begin{figure}[h]
\begin{center}
\begin{tikzpicture}[>=latex]
\node (s0) at (0.7,9) {\scriptsize $[\oC_1,\oC_2,\oC_3,\oC_4]$};
\node (s1) at (-2,8) {\scriptsize $[\oC_1,\{3,4\},\oC_3,\oC_4]$};
\node (s11) at (-4.5,7) {\scriptsize $[\oC_1,\{4\},\oC_3,\oC_4]$};
\node (s112) at (-3.1,6) {\scriptsize $[\oC_1{,}\es{,}\oC_3{,}\oC_4]$};
\node (s1121) at (-3.2,5) {\scriptsize $[\oC_1{,}\es{,}\oC_3{,}\{6,7\}]$};
\node (s1122) at (-1.1,5) {\scriptsize $[\oC_1{,}\es{,}\{4,6\}{,}\oC_4]$};
\node (s11221) at (-1.4,4) {\scriptsize $\cvr_4$};
\node (s11222) at (-.4,4) {\scriptsize $\cvr_5$};
\draw[->] (s1122) -- node[left]{\scriptsize $(6,3)$} (s11221);
\draw[->] (s1122) -- node[right]{\scriptsize $(6,4)$} (s11222);
\node (s11211) at (-3.8,4) {\scriptsize $\cvr_2$};
\node (s11212) at (-2.7,4) {\scriptsize $\cvr_3$};
\draw[->] (s11) -- node[right]{\scriptsize $(4,3)$} (s112);
\draw[->] (s112) -- node[left]{\scriptsize $(5,3)$} (s1121);
\draw[->] (s112) -- node[right]{\scriptsize $(5,4)$} (s1122);
\draw[->] (s1121) -- node[left,near start]{\scriptsize $(6,3)$} (s11211);
\draw[->] (s1121) -- node[right]{\scriptsize $(6,4)$} (s11212);
\node (s111) at (-5.2,6) {\scriptsize $[\oC_1{,}\{4\}{,}\{5,6\}{,}\oC_4]$};
\node (s1112) at (-5.3,5) {\scriptsize $[\oC_1{,}\{4\}{,}\{6\}{,}\oC_4]$};
\node (s11122) at (-5.3,4) {\scriptsize $\cvr_1$};
\draw[->] (s1112) -- node[left,near start]{\scriptsize $(6,4)$} (s11122);
\node (s12) at (.8,7) {\scriptsize $[\{1,2\}{,}\{3,4\}{,}\oC_3{,}\oC_4]$};
\node (s122) at (.8,6) {\scriptsize $[\{1,2\}{,}\{3,4\}{,}\{5,6\}{,}\oC_4]$};
\node (s1222) at (1.3,5) {\scriptsize $[\{1,2\}{,}\{3,4\}{,}\{6\}{,}\oC_4]$};
\node (s12222) at (1.3,4) {\scriptsize $\cvr_7$};
\draw[->] (s1222) -- node[right]{\scriptsize $(6,4)$} (s12222);
\draw[->] (s122) -- node[right]{\scriptsize $(5,4)$} (s1222); 
\draw[->] (s12) -- node[right]{\scriptsize $(4,2)$} (s122); 
\node (s2) at (4.6,8) {\scriptsize $[\{1,3\},\oC_2,\oC_3,\oC_4]$};
\node (s21) at (3.9,7) {\scriptsize $[\{1,3\},\{2,4\},\oC_3,\oC_4]$};
\node  (s212) at (3.9,6) {\scriptsize $[\{1{,}3\}{,}\{2{,}4\}{,}\{5{,}6\},\oC_4]$};
\node (s2122) at (4,5) {\scriptsize $[\{1,3\}{,}\{2,4\}{,}\{6\},\oC_4]$};
\node (s21222) at (4,4) {\scriptsize $\cvr_9$};
\draw[->] (s2122) -- node[right]{\scriptsize $(6,4)$} (s21222);
\draw[->] (s212) -- node[right]{\scriptsize $(5,4)$} (s2122);
\draw[->] (s21) -- node[right]{\scriptsize $(4,2)$}(s212);
\node (s22) at (6.5,7) {\scriptsize $[\{1\}{,}\oC_2{,}\oC_3{,}\oC_4]$};

\node (mak1) at (6.5,6) {\scriptsize $[\{1\}{,}\oC_2{,}\{5{,}6\}{,}\oC_4]$};
\node (mak2) at (6.5,5) {\scriptsize $[\{1\}{,}\oC_2{,}\{6\}{,}\oC_4]$};
\node (mak3) at (6.5,4) {\scriptsize $\cvr_{11}$};
\draw[->] (mak1) -- node[right]{\scriptsize $(5,4)$}(mak2);
\draw[->] (mak2) -- node[right]{\scriptsize $(6,4)$}(mak3);
\draw[->] (s22) -- node[right,near start]{\scriptsize $(4,2)$}(mak1);
\draw[->] (s2) -- node[left]{\scriptsize $(3,1)$}(s21);
\draw[->] (s2) -- node[right]{\scriptsize $(3,2)$}(s22);
\draw[->] (s0) -- node[sloped,above,near end]{\scriptsize $(2,1)$} (s1);
\draw[->] (s1) -- node[sloped,above,near end]{\scriptsize $(3,1)$} (s11);
\draw[->] (s11) -- node[left]{\scriptsize $(4,2)$} (s111);
\draw[->] (s111) -- node[left]{\scriptsize $(5,4)$} (s1112);
\draw[->] (s1) -- node[sloped,above]{\scriptsize $(3,2)$} (s12);
\draw[->] (s0) -- node[sloped,above]{\scriptsize $(2,2)$} (s2);
\end{tikzpicture}
\end{center}
where $\cvr_1=[\oC_1{,}\{4\}{,}\es{,}\oC_4]$, $\cvr_2=[\oC_1{,}\es{,}\oC_3{,}\{7\}]$, $\cvr_3=[\oC_1{,}\es{,}\{4,5\}{,}\{6,7\}]$, $\cvr_4=[\oC_1{,}\es{,}\{4,6\}{,}\{5,7\}]$, $\cvr_5=[\oC_1{,}\es{,}\{4\}{,}\oC_4]$, $\cvr_6=[\{1,2\}{,}\{3\}{,}\oC_3{,}\oC_4]$, $\cvr_7=[\{1,2\}{,}\{3,4\}{,}\es{,}\oC_4]$, $\cvr_8=\{\{1,3\}{,}\{2\}{,}\oC_3,{,}\oC_4\}$, $\cvr_9=[\{1,3\}{,}\{2,4\}{,}\es{,}\oC_4]$, $\cvr_{10}=[\{1\},\{2,3\},\oC_3,\oC_4]$, $\cvr_{11}=[\{1\},\oC_2,\es,\oC_4]$.
\caption{Search tree $\tT^S_1(G)$ for the undirected graph from Example \ref{ex:dumbo}.}
\label{fig:dt1}
\end{figure}
 \end{example}
\section{Avoiding repeated computations of the same maximal clique-partition}
\label{sect:norepeat}
In Example \ref{ex:dumbo},  the final configurations $\cvr_1:=[\oC_1{,}\{4\}{,}\es{,}\oC_4]$ and $\cvr_{5}:=[\oC_1{,}\es{,}\{4\}{,}\oC_4]$ represent the same maximal clique-partition: 
$\repr(\cvr_1)=\repr(\cvr_5)=\{\{1,2,3\},\{4\},\{5,6,7\}\}$.
Thus, there are situations when the search space $\tT^S_{!(T1)}(G)$ has the following undesirable feature: different final configurations represents the same maximal clique-partition. This implies that some computations are redundant: some maximal clique-partitions will be generated more than once.

Note that, in Example \ref{ex:dumbo}, the configurations  $\cvr_1$ and $\cvr_5$ which represent the same maximal clique-partition $\cP=\{\oC_1,\{4\},\oC_4\}$ have the following property: $\cvr_1=[C_1,\ldots,C_m]$, $\cvr_5=[C'_1,\ldots,C'_m]$ and there exist $1\leq j\neq i\leq m$ such that $C_i=C'_j\neq\es$ and $C'_i=\es=C_j$. The following lemma indicates that this is a general property of configurations which represent the same maximal clique-partition:
\begin{lemma}
\label{lema4}
If the distinct final configurations $\cvr=[C_1,\ldots,C_m],\cvr'=[C'_1,\ldots,C'_m]$ in $\tT_{!(T1)}^S(G)$ have  $\repr(\cvr)=\repr(\cvr')$
then there  exist $1\leq j\neq i\leq m$ such that $C_i=C'_j\neq\es$ and $C'_i=C_j=\es$.
\end{lemma}
\begin{proof}
Let $I:=\{k\mid 1\leq k\leq m$ and $C_k\neq\emptyset\}$. Then $\repr(\cvr)=\{C_k\mid k\in I\}$. Moreover,
\begin{itemize}
\item $\repr(\cvr)=\repr(\cvr')$ implies the existence of a permutation $\pi:\{1,\ldots,m\}\to\{1,\ldots,m\}$ such that $C'_k=C_{\pi(k)}$ for all $k\in [m]$, 
\item $\cvr\neq\cvr'$ implies the existence of $i\in I$ such that $C_i\neq C'_i$. This implies $i\neq \pi(i)$. 
\end{itemize}
Let $j=\pi^{-1}(i)$. Then $j\neq i$ and $C_i=C_{\pi(j)}=C'_j$. Since $i\in I$, we have $C_i=C'_j\neq\es$.

Note that $j\neq \pi(j)$ because $j=\pi(j)$ implies $j=\pi(\pi^{-1}(i))=i$, contradiction.

%Then we have the situation depicted below:
%$$\begin{array}{|rlllllll|}\hline
% \cvr_0= [&\!\!\!\!\oC_1,&\ldots&\oC_j,&\ldots&\oC_i,&\ldots&\oC_m],\\ \hline
%\cvr=[&\!\!\!\!C_1, &\ldots&  C_j,&\ldots&C_i,&\ldots&C_m],\\ \hline
%\cvr'=[&\!\!\!\!C'_1,&\ldots&C'_j,&\ldots & C'_i,&\ldots&C'_m]\\
%=[&\!\!\!\!C_{\pi(1)}, &\ldots&C_{\pi(j)},&\ldots&C_{\pi(i)},&\ldots&C_{\pi(m)}]\\ 
%=[&\!\!\!\!C_{\pi(1)}, &\ldots&C_i,&\ldots&C_{\pi(i)},&\ldots&C_{\pi(m)}]\\ \hline
%\end{array}$$

It remains to prove that $C'_i=C_j=\es$. From $\es\neq C_i\sse \oC_i$ and $C_{\pi(i)}=C'_i\sse\oC_i$ we learn that $C_i\cup C_{\pi(i)}$ is a clique included in $\oC_i$. 
We must have $C'_i=C_{\pi(i)}=\es$ because otherwise $C_i$ and $C_{\pi(i)}$ would be different cliques of $\repr(\cvr)$ with $C_j\cup C_{\pi(j)}$ a clique, 
which contradicts the assumption that $\repr(\cvr)$ is maximal clique-partition. 

From $\es\neq C'_j=C_{\pi(j)}\sse\oC_j$ and $C_j\sse\oC_j$ we learn that $C_{\pi(j)}\cup C_j$ is a clique included in $\oC_j$.
We must have $C_j=\es$ because otherwise 
$C_j$ and $C_{\pi(j)}$ would be different cliques of $\repr(\cvr)$ with $C_j\cup C_{\pi(j)}$ a clique, which contradicts the assumption that $\repr(\cvr)$ is maximal clique-partition. 
\end{proof}
Lemma \ref{lema4} allows us to define a criterion to eliminate from $\tT^S_{!(T1)}(G)$ the redundant computations of maximal clique-partitions.

\begin{definition}
We say that configuration $\cvr=[C_1,\ldots,C_m]\in\tT^S_{!(T1)}(G)$ with $\delta(\cvr)=[i_1,\ldots,i_\ell]$ is a {\bf $T2$-node}, or that it has property $T2$, if there exist $1\leq j<i\leq m$ such that $j\in\{i_1,\ldots,i_\ell\}-\Rgd$ and $i\in\cliques(C_j)$. We let $\tT^S_{!(T_1,T_2)}(G)$ be 
be the result of pruning from $\tT^S_{!(T1)}(G)$ all subtrees whose root is a $T2$-node. 
\end{definition}
From now on we write  $\leaf(\tT_{!(T1,T2})^S(G))$ for the set of configurations in $\tT_{!(T1,T2)}^S(G)$ at depth $s$.
If we let $\Pi^!_{\mathit{mcp}}(G)$ be the set of configurations $[C_1,\ldots,C_m]$ from $\Pi_{\mathit{mcp}}(G)$ such that 
\begin{center}
for all $1\leq j<i\leq m$, if $C_j\neq \es$ then $i\not\in\cliques(C_j)$
\end{center}
then the following lemmas hold:
\begin{lemma}
\label{lema6}
For every maximal clique-partition $\cC$ of $G$ there is exactly one configuration $\cvr\in \Pi^!_{\mathit{mcp}}(G)$ with $\repr(\cvr)=\cC$. 
\end{lemma}
\begin{proof}
For every  configuration $\cvr=[C_1,\ldots,C_m]\in\leaf(\tT^S_{!(T1)}(G))$ we define the measure $$m(\cvr):=\sum_{\substack{1\leq i\leq m\\
C_i=\es}}i.$$
First, we prove that $\Pi^!_{mcp}(G)$ has at least one configuration whose representation is $\cP$. By Lemma \ref{mcp-cp}, there exists $\cvr'_0=[C_1,\ldots,C_m]\in \Pi_{mcp}(G)$ with $\repr(\cvr)=\cP$.
If $\cvr\in\Pi^!_{mcp}(G)$ then we are done. Otherwise there exists $1\leq j<i\leq m$ such that $\es\neq C_j\sse\oC_i$. Then $C_i=\es$ because otherwise $C_i,C_j$ would be different components of $\cP$ with $C_i\cup C_j\sse \oC_i$ and this contradicts the assumption that the clique-partition $\cP$ is maximal.
Thus, we can define the configuration $\cvr'_1=[C'_1,\ldots,C'_m]\in\Pi_{\mathit{mcp}}(G)$ with
$$C'_k:=\left\{\begin{array}{ll}
C_j&\text{if }k=i,\\
\es&\text{if }k=j,\\
C_k&\text{otherwise}
\end{array}\right.$$
If follows that $m(\cvr'_0)-m(\cvr'_1)=i-j>0$. In this way we can build a sequence of configurations $\cvr'_0,\cvr'_1,\ldots\in\Pi_{mcp}(G)$ with $\cP=\repr(\cvr'_0)=\repr(\cvr'_1)=\ldots$
and $m(\cvr'_0)>m(\cvr'_1)>\ldots$\ 
Since the ordering $>$ on natural numbers is well-founded, this construction will eventually end with a configuration $\cvr'_p\in\Pi_{\mathit{mcp}}^!(G)$ and $\repr(\cvr'_p)=\cP$. Hence $\Pi_{\mathit{mcp}}^!(G)$ has at least one configuration whose representation is $\cP$.

It remains to show that there are no two configurations $\cvr=[C_1,\ldots,C_m], \cvr'=[C'_1,\ldots,C'_m]\in\Pi_{\mathit{mcp}}^!(G)$ with $\repr(\cvr)=\repr(\cvr')$. If this were the case then, by Lemma \ref{lema4}, there exist $1\leq j\neq i\leq m$ such
that $C_i=C'_j\neq\es$ and $C'_i=C_j=\es$.
If $j<i$ then $\es\neq C'_j=C_i\sse\oC_i$, which contradicts the assumption that $\cvr'\in\Pi_{mcp}^!(G)$. If $j>i$ then $\es\neq C_i=C'_j\sse\oC_j$, which contradicts the assumption that $\cvr\in\Pi_{\mathit{mcp}}^!(G)$. 
\end{proof}

\begin{lemma}
\label{lemi3}
$\leaf(\tT^S_{!(T1,T2)}(G))=\Pi_{\mathit{mcp}}^!(G)$.
\end{lemma}
\begin{proof}
First, we prove that $\leaf(\tT^S_{!(T1,T2)}(G))\sse\Pi_{\mathit{mcp}}^!(G)$. Let  $\cvr=[C_1,\ldots,C_m]\in \leaf(\tT^S_{!(T1,T2)}(G))$ with $\delta(\cvr)=[i_1,\ldots,i_s]$. Since $\leaf(\tT^S_{!(T1,T2)}(G))\sse \leaf(\tT^S_{!(T1)}(G))=\Pi_{\mathit{mcp}}(G)$, we only have to show that, for all $1\leq j<i\leq m$, if $C_j\neq\es$ then $i\not\in\cliques(C_j)$. If this were not the case, then there exist $j\in\{i_1,\ldots,i_s\}\cup\Rgd$ and $j< i\leq m$ such that $i\in\cliques(C_j)$. We observe that $j\not\in\Rgd$ because otherwise $\cliques(C_j)=\{j\}$, which contradicts the assumption $i\in \cliques(C_j)$. This implies that $[C_1,\ldots,C_m]$ is $T2$-node of $\tT_{!(T1)}^S(G)$, which contradicts the assumption that $[C_1,\ldots,C_m]\in\leaf(\tT_{!(T1,T2)}^S(G))$.

To finish the proof, we must show that $\Pi_{mcp}^!(G)\sse\leaf(\tT^S_{!(T1,T2)}(G))$. Since $$\Pi_{\mathit{mcp}}^!(G)\sse\Pi_{\mathit{mcp}}(G)=\leaf(\tT^S_{!(T1)}(G)),$$ it is sufficient to prove that every configuration $\cvr\in\Pi_{\mathit{mcp}}(G)-\Pi_{\mathit{mcp}}^!(G)$ is below a $T2$-node of $\tT^S_{!(T1)}(G)$. Let $\cvr=[C_1,\ldots,C_m]$ with $\delta(\cvr)=[i_1,\ldots,i_s]$. Then there exist $1\leq j<i<m$ such that $\es\neq C_j\sse \oC_i$. We observe that $j\not\in \Rgd$ because otherwise the only maximal clique which contains $C_j$ is $\oC_j$. Therefore we must have $j\in\{i_1,\ldots,i_s\}-\Rgd$. This implies that  $\cvr$ is $T2$-node. 
\end{proof}
The following corollary is an immediate consequence of the previous two lemmas.
\begin{corollary}
Every maximal clique-partition is produced by a single configuration from $\leaf(\tT^S_{!(T1,T2)}(G))$.
\end{corollary}
\begin{example}
The tree $\tT^S_{!(T1,T2)}(G)$ for the graph $G$ from Example \ref{ex:dumbo} is shown in Figure \ref{fig:dt2}. \\
$[\oC_1{,}\{4\}{,}\{5,6\}{,}\oC_4]\in\tT^S_{!(T1)}(G)$ is a $T2$-node because it is of the form $[C_1,C_2,C_3,C_4]$ and there exist $j=2<3=i$ such that $j\in\{2,3\}-\Rgd$ and $C_2\sse\oC_3$.
Compared to $\tT^S_{!(T1)}(G)$, the total number of non-root nodes in $\tT^S_{!(T1,T2)}(G)$ has dropped from 25 to 22, and $$\leaf(\tT^S_{!(T1,T2)}(G))=\Pi^!_{\mathit{mcp}}(G)=\{\cvr_2, \cvr_3, \cvr_4, \cvr_5, \cvr_{7}, \cvr_{9}, \cvr_{11}\}.$$ Every maximal clique-partition is produced by a single final configuration in $\tT^S_{!(T1,T2)}(G))$:
$$\begin{array}{rr}
\repr(\cvr_2)=\{\{1,2,3\},\{4,5,6\},\{7\}\},&\repr(\cvr_3)=\{\{1,2,3\},\{4,5\},\{6,7\}\},\\
\repr(\cvr_4)=\{\{1,2,3\},\{4,6\},\{5,7\}\},&\repr(\cvr_5)=\{\{1,2,3\},\{4,5,6\},\{7\}\},\\
\repr(\cvr_7)=\{\{1,2\},\{3,4\},\{5,6,7\}\},&\repr(\cvr_9)=\{\{1,3\},\{2,4\},\{5,6,7\}\},\\
\repr(\cvr_{11})=\{\{1\},\{2,3,4\},\{5,6,7\}\}.
\end{array}$$
\end{example}

\begin{figure}[ht]
\begin{center}
\begin{tikzpicture}[>=latex,scale=1.1]
\node (s0) at (1,9) {\scriptsize $[\oC_1,\oC_2,\oC_3,\oC_4]$};
\node (s1) at (-1.5,8) {\scriptsize $[\oC_1,\{3,4\},\oC_3,\oC_4]$};
\node (s11) at (-2.6,7) {\scriptsize $[\oC_1,\{4\},\oC_3,\oC_4]$};
\node (s112) at (-2.6,6) {\scriptsize $[\oC_1{,}\es{,}\oC_3{,}\oC_4]$};
\node (s1121) at (-3.6,5) {\scriptsize $[\oC_1{,}\es{,}\oC_3{,}\{6,7\}]$};
\node (s1122) at (-1.4,5) {\scriptsize $[\oC_1{,}\es{,}\{4,6\}{,}\oC_4]$};
\node (s11221) at (-1.6,4) {\scriptsize $\cvr_4$};
\node (s11222) at (-.6,4) {\scriptsize $\cvr_5$};
\draw[->] (s1122) -- node[left,near start]{\scriptsize $(6,3)$} (s11221);
\draw[->] (s1122) -- node[right,near start]{\scriptsize $(6,4)$} (s11222);
\node (s11211) at (-4,4) {\scriptsize $\cvr_2$};
\node (s11212) at (-3.2,4) {\scriptsize $\cvr_3$};
\draw[->] (s11) -- node[right]{\scriptsize $(4,3)$} (s112);
\draw[->] (s112) -- node[left]{\scriptsize $(5,3)$} (s1121);
\draw[->] (s112) -- node[right]{\scriptsize $(5,4)$} (s1122);
\draw[->] (s1121) -- node[left,near start]{\scriptsize $(6,3)$} (s11211);
\draw[->] (s1121) -- node[right,near start]{\scriptsize $(6,4)$} (s11212);
% CUT
\node (s12) at (.9,7) {\scriptsize $[\{1,2\},\{3,4\},\oC_3,\oC_4]$};
\node (s122) at (.9,6) {\scriptsize $[\{1,2\}{,}\{3,4\}{,}\{5,6\}{,}\oC_4]$};
\node (s1222) at (.9,5) {\scriptsize $[\{1,2\}{,}\{3,4\}{,}\{6\}{,}\oC_4]$};
\node (s12222) at (.9,4) {\scriptsize $\cvr_7$};
\draw[->] (s1222) -- node[right,near start]{\scriptsize $(6,4)$} (s12222);
\draw[->] (s122) -- node[right,near start]{\scriptsize $(5,4)$} (s1222); 
\draw[->] (s12) -- node[right,near start]{\scriptsize $(4,2)$} (s122); 
\node (s2) at (4,8) {\scriptsize $[\{1,3\},\oC_2,\oC_3,\oC_4]$};
\node (s21) at (3.8,7) {\scriptsize $[\{1,3\},\{2,4\},\oC_3,\oC_4]$};
\node (s212) at (3.8,6) {\scriptsize $[\{1,3\}{,}\{2,4\}{,}\{5,6\},\oC_4]$};
\node (s2122) at (3.8,4.9) {\scriptsize $[\{1,3\}{,}\{2,4\}{,}\{6\},\oC_4]$};
\node (s21222) at (3.8,4) {\scriptsize $\cvr_9$};
\draw[->] (s2122) -- node[left,near start]{\scriptsize $(6,4)$} (s21222);
\draw[->] (s212) -- node[left,near start]{\scriptsize $(5,4)$} (s2122);
\draw[->] (s21) -- node[right]{\scriptsize $(4,2)$}(s212);
\node (s22) at (6.3,7) {\scriptsize $[\{1\}{,}\oC_2{,}\oC_3{,}\oC_4]$};
\node (s221) at (6.3,6) {\scriptsize $[\{1\}{,}\oC_2{,}\{5,6\}{,}\oC_4]$};
\node (s2211) at (6.3,4.9) {\scriptsize $[\{1\}{,}\oC_2{,}\{6\}{,}\oC_4]$};
\node (y) at (6.3,4) {\scriptsize $\cvr_{11}$};
\draw[->] (s22) -- node[left]{\scriptsize $(4,2)$} (s221);
\draw[->] (s221) -- node[left,near start]{\scriptsize $(5,4)$} (s2211);
\draw[->] (s2211) -- node[left,near start]{\scriptsize $(6,4)$} (y);
\draw[->] (s2) -- node[left]{\scriptsize $(3,1)$}(s21);
\draw[->] (s2) -- node[right]{\scriptsize $(3,2)$}(s22);
\draw[->] (s0) -- node[sloped,above]{\scriptsize $(2,1)$} (s1);
\draw[->] (s1) -- node[left,near start]{\scriptsize $(3,1)$} (s11);
\draw[->] (s1) -- node[right,near start]{\scriptsize $(3,2)$} (s12);
\draw[->] (s0) -- node[sloped,above]{\scriptsize $(2,2)$} (s2);
\end{tikzpicture}
\end{center}
The configurations $\cvr_i$ for $2\leq i\leq 11$ are the same as those from Figure \ref{fig:dt1}.
\caption{The search tree $\tT^S_{!(T1,T2)}(G)$ for the graph from Example \ref{ex:dumbo}.}
\label{fig:dt2}
\end{figure}

\section{The combined detection of $T1$-nodes and $T2$-nodes}
\label{sect:combo}
Let $\cvr \to_{(v_\ell,i_\ell)} \cvr'$ be a branch of $\tT^S_{!(T1,T2)}(G)$, and $\delta(\cvr')=[i_1,\ldots,i_\ell]$. 
The only visible change from $\cvr$ to  $\cvr'$ is that of the components with indices from the set $\cliques(v_\ell)-\{i_\ell\}$.
When we are about to check if $\cvr'$ has property $T1$ or $T2$, we know that $\cvr$ does not have property $T1$ or $T2$ because it has already passed this test.
Therefore, it is sufficient to consider only the set of clique indices $\cJ_\ell:=\cliques(v_\ell)\cap\cI_\ell$ where $\cI_\ell:=\{i_1,\ldots,i_\ell\}-Rgd$, and to check if $\cvr'$ is a configuration $[C_1,\ldots,C_m]$ which satisfies one of the following conditions for some $j\in\cJ_\ell$: 
\begin{enumerate}
 \item (a) $C_j\sse\oC_i$ for an $i\in Rgd$ (in this case $\cvr'$ is $T1$-node), or (b) $i>j$ (in this case $\cvr'$ is $T2$-node); or
 \item there exists $i\in\cI_\ell-\{j\}$ such that $C_i\cup C_j$ is a clique. In this case $\cvr'$ is $T1$-node.
\end{enumerate}
Condition 1.(a) is equivalent with $\cliques(C_j)\cap\Rgd\neq\es$, and condition 2 is equivalent with $\cliques(C_i)\cap\cliques(C_j)\neq\es$.

\section{Enumerating all maximal clique-partitions}
\label{sect:algo}
In this section we describe an algorithm to compute the maximal clique-partitions one-by-one, on request.
The following global data is assumed to be available:
\begin{itemize}
\item $\oC_1,\ldots,\oC_m$: the maximal cliques of  $G$
\item $\cliques(v)=\{k\in[m]\mid v\in \oC_k\}$ for all $v\in V$
\item $\Rgd=\{k\mid \cliques(v)=\{k\}$ for some $v\in V\}$
\item $S=[v_1,\ldots,v_s]$: an enumeration of all vertices $v\in V$ with $d(v)>1$
\end{itemize}
During the computation we will keep track of the following information:
\begin{itemize}
\item $\ell$: the depth of the search in tree $\tT_{!(T1,T2)}^S(G)$
\item $\cvr[\ell]$: the current configuration of the search in tree $\tT_{!(T1,T2)}^S(G)$
\item $choice[i]$ for $1\leq i\leq \ell$: the index of the clique where vertex $v_i\in S$ is assigned. If $choice[i]==0$ then vertex $v_i$ was not yet been assigned to any clique. 
\end{itemize}

\begin{tabbing}
{\sc pr}\={\sc ocedure} {\tt initSearch()}\\
\>$\cvr[1]:=[\oC_1,\ldots,\oC_m]$;\\
\>$\ell:=1$;\\
\>{\bf for} \=$i:=1$ to $s$ {\bf do} \\
\>\>$choice[i]:=0$;\\
\>{\tt findNextClique()};\\ \\
{\sc pr}\={\sc ocedure} {\tt hasNext()}\\
\>{\bf return} $\cvr[1]\neq\tt null$\\
\\
{\sc pr}\={\sc ocedure} {\tt next()}\\
\>{\bf if} \=($\cvr[1]==\tt null$)
 {\bf return} {\tt null};\\
\>{\bf else}\\
\>\> $result:=\repr(\cvr[\ell])$;\\
\>\> {\tt findNextClique()};\\
\>\> {\bf return} $result$;\\ \\
{\sc pr}\={\sc ocedure} {\tt findNextClique()} \\
\>{\bf if} \=$s>0$\\
\>\>{\bf wh}\={\bf ile} $\ell\geq 1$\\
\>\>\>$V_\ell:=\{k\in\cliques(v_\ell)\mid$\,\=$k>choice[\ell]$ and $(k\in\Rgd\text{ or }\cliques(C_k)\cap\Rgd=\es)\}$;\\
\>\>\>{\bf if} \=$V_\ell$ is empty\\
\>\>\>\>$choice[\ell]:=0$;\\
\>\>\>\>$\ell:=\ell-1$;\\
\>\>\>{\bf else}\\
\>\>\>\>$i:=\min V_\ell$; \\
\>\>\>\>// keep $v_\ell$ only in clique with index $i$\\
\>\>\>\>$choice[\ell]:=i$;\\
\>\>\>\>$\cvr[\ell]:=[C'_1,\ldots,C'_m]$ where $C'_k:=\left\{\begin{array}{ll}
C_k-\{v_\ell\}&\text{if }k\in\cliques(v_\ell)-\{i\},\\
C_k&\text{otherwise}
\end{array}\right.$ \\
%\>\>\>\>{\bf for} \=all $k\in\cliques(v_\ell)$\\
%\>\>\>\>\> {\bf if} $(k\neq i)$ {\bf then} $C_k:=C_k-\{v_\ell\}$\\
\>\>\>\>{\bf if}\=\ ({\tt isT1OrT2($\cvr[\ell]$)})
{\bf continue};\\
\>\>\>\>{\bf el}\={\bf se}\\
\>\>\>\>\>{\bf if} \=$(\ell==s)$\\
\>\>\>\>\>\> {\bf return};\qquad\qquad  // maximal clique-partition detected\\
\>\>\>\>\>{\bf else}\\
\>\>\>\>\>\>$\cvr[\ell+1]:=\cvr[\ell]$;\\
\>\>\>\>\>\>$\ell:=\ell+1$;\\
\>$\cvr[1]:=\tt null$;\\ \\
{\sc pr}\={\sc ocedure} {\tt isT1orT2($[C_1,\ldots,C_m]$)}\\
\>$I:=\{choice[k]\mid 1\leq k\leq \ell\}-\Rgd$;\\
\>$J:= I \cap\cliques(v_\ell)$;\\
\>{\bf for} all $j\in J$\\
\>\>{\bf if} $\cliques(C_j)\cap\Rgd\neq\es$ or $\cliques(C_j)\cap\{i\mid j<i\leq m\}\neq\es$\\
\>\>\>{\bf return} {\tt true};\\
\>\>{\bf for} all $i\in I$\\
\>\>\>{\bf if} \=$i\in J$\\
\>\>\>\>{\bf if} $(j<i)$ and $\cliques(C_i)\cap\cliques(C_j)\neq\es$\\
\>\>\>\>\>{\bf return} {\tt true};\\
\>\>\>{\bf else if} $\cliques(C_i)\cap\cliques(C_j)\neq\es$\\
\>\>\>\>{\bf return} {\tt true};\\
\>{\bf return} {\tt false};
\end{tabbing}

%%
%% Bibliography
%%

%% Please use bibtex,
\section{Complexity}
\label{sect:complexity}
\begin{theorem}
	If the set of vertices of $G$ is $\{v_1,\ldots, v_n \}$ then the number of maximal clique-partitions of $G$ is at most $ \prod_{i=1}^{n} d(v_i)$.  
\end{theorem}
\begin{proof}
The result directly follows from the facts that (1) every maximal clique-partition is produced by some final configuration of $\tT^S(G)$, and (2)
$\tT^S(G)$ has $\prod_{i=1}^{n} d(v_i)$ final configurations.
%shared vertices are distributed in their containing cliques in all possible ways. 
\end{proof}

It is easy to see that this upper bound can be reached. Just consider the graph with two maximal cliques: $C_1=\{p_1,\ldots,p_n, \mathit{true} \}$ and $C_2=\{p_1,\ldots,p_n, \mathit{false} \}$. The set of all maximal clique-partitions imitates the truth assignment in propositional logic, containing $2^n$ maximal clique-partitions.

This theorem implies that the algorithm is exponential in the number of vertices shared among multiple cliques. On the other hand, the length of each branch of the algorithm is polynomially bounded, since it requires at most as many steps as there are vertices shared among maximal cliques. Therefore, every single maximal clique-partition can be computed in polynomial time.
\paragraph*{Experimental results}
To test the performance of our algorithm, we implemented it in Java and Mathematica~\cite{DBLP:books/daglib/0011324}, and ran it on a MacBook Air M2
with 8-core CPU and 8 GB  RAM. We indicate the runtimes to enumerate all maximal clique-partitions of some graphs from the following families:
\begin{enumerate}
\item $G_n$, $n\geq 2$, obtained by extending the complete graph $K_n$ with $n$ new vertices, and connecting every vertex of $K_n$ with a distinct new vertex.
Examples of graphs $G_n$ are:
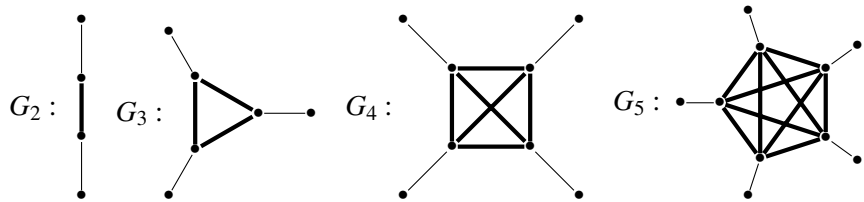
\begin{center}
\begin{tikzpicture}[inner sep=0pt,scale=1.3]
\node (g0) at (0,0.3) {\scriptsize $\bullet$};
\node (g1) at (0,-.3) {\scriptsize $\bullet$};
\node (g2) at (0,.9) {\scriptsize $\bullet$};
\node (g3) at (0,-.9) {\scriptsize $\bullet$};
\draw[ultra thick] (g0) -- (g1);
\draw (g0) -- (g2);
\draw (g1) -- (g3);
\node[left] at (-.25,0) {$G_2:$};
\end{tikzpicture}\quad
\begin{tikzpicture}[scale=1.4,inner sep=0pt]
\node[left] at (-.5,0) {$G_3:$};
\node (0) at (.4,0) {\scriptsize $\bullet$};
\node (1) at (-.2,0.346) {\scriptsize $\bullet$};
\node (2) at (-.2,-.346) {\scriptsize $\bullet$};
\node (3) at (.9,0) {\scriptsize $\bullet$};
\node (4) at (-.45,0.78) {\scriptsize $\bullet$};
\node (5) at (-.45,-.78) {\scriptsize $\bullet$};
\draw[ultra thick] (0) -- (1) -- (2) -- (0);
\draw (0) -- (3);
\draw (1) -- (4);
\draw (2) -- (5);
\end{tikzpicture}\quad
\begin{tikzpicture}[inner sep=0pt,scale=1.3]
\node[left] at (-1,0) {$G_4:$};
\node (0) at (.4,.4) {\scriptsize $\bullet$};
\node (1) at (.4,-.4) {\scriptsize $\bullet$};
\node (2) at (-.4,-.4) {\scriptsize $\bullet$};
\node (3) at (-.4,.4) {\scriptsize $\bullet$};
\node  (4) at (.9,.9) {\scriptsize $\bullet$};
\node  (5) at (.9,-.9) {\scriptsize $\bullet$};
\node (6) at (-.9,-.9) {\scriptsize $\bullet$};
\node (7) at (-.9,.9) {\scriptsize $\bullet$};
\draw[ultra thick] (0) -- (1) -- (2) -- (3) -- (0) -- (2) (1) -- (3);
\draw (0) -- (4);
\draw (1) -- (5);
\draw (2) -- (6);
\draw (3) -- (7);
\end{tikzpicture}\quad
\begin{tikzpicture}[inner sep=0pt,scale=1.3]
\node[left] at (-1.2,0) {$G_5:$};
\node (0) at (-.6,0) {\scriptsize $\bullet$};
\node (1) at (-.185,-.57) {\scriptsize $\bullet$};
\node (2) at (.485,-.353) {\scriptsize $\bullet$};
\node (3) at (.485,0.353) {\scriptsize $\bullet$};
\node (4) at (-.185,.57) {\scriptsize $\bullet$};
\node (x0) at (-1,0) {\scriptsize $\bullet$};
\node (x1) at (-.309,-.95) {\scriptsize $\bullet$};
\node (x2) at (.81,-.588) {\scriptsize $\bullet$};
\node (x3) at (.81,0.588) {\scriptsize $\bullet$};
\node (x4) at (-.309,.95) {\scriptsize $\bullet$};
\draw (0) -- (x0) (1) -- (x1) (2) -- (x2) (3) -- (x3) (4) -- (x4);
\draw[ultra thick] (0) -- (1) -- (2) -- (3) -- (4) -- (0) -- (2) -- (4) -- (1) -- (3) -- (0);
\end{tikzpicture}
\end{center}

\begin{center}
\begin{tabular}{|c|c|r|r|r|} \hline
&&Number of& Number of maximal&\\
Graph&Order &vertices in $S$& clique-partitions&Runtime\\ \hline
$G_4$&8&4 &12&0.002 \\
$G_5$&10&5 &27&0.003 \\
$G_6$&12&6 &58&0.005 \\
$G_7$&14&7 &121&0.007 \\
$G_8$&16&8 &248&0.012 \\
$G_{10}$&20&10 &1014&0.270 \\
$G_{20}$&40&20 &1048556&2.462 \\ 
$G_{25}$&50&25 &33554407&61.706 \\ 
\hline
\end{tabular}
\end{center}
\item $H_n$ of order $4\,n$ ($n\geq 2$), with three maximal cliques of order $2\,n$: two  are mutually disjoint, and the third one shares $n$ vertices with each of the other two maximal cliques.
For example:
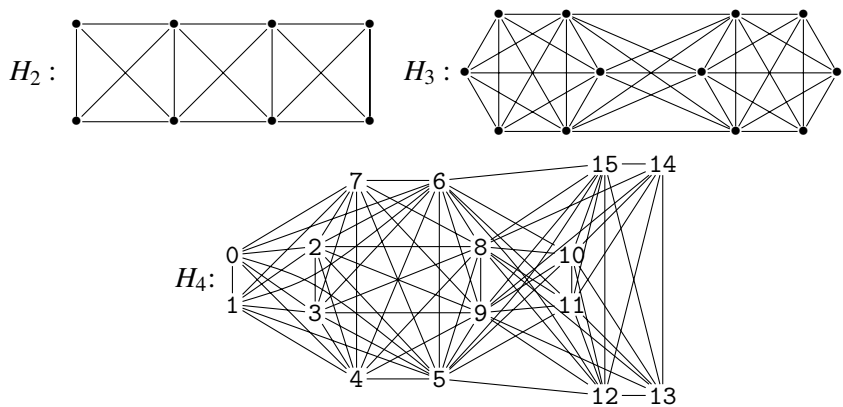
\begin{center}
\begin{tikzpicture}[inner sep=0pt,baseline=0pt,scale=1.3]
\node (0) at (-.5,.5) {\scriptsize $\bullet$};
\node (1) at (.5,.5) {\scriptsize 
$\bullet$};
\node (2) at (1.5,.5) {\scriptsize 
$\bullet$};
\node (3) at (2.5,.5) {\scriptsize 
$\bullet$};
\node (4) at (-.5,-.5) {\scriptsize $\bullet$};
\node (5) at (.5,-.5) {\scriptsize 
$\bullet$};
\node (6) at (1.5,-.5) {\scriptsize 
$\bullet$};
\node (7) at (2.5,-.5) {\scriptsize 
$\bullet$};
\draw (0) -- (1) -- (2) -- (3) -- (7) -- (6) -- (5) -- (4) -- (0) (1) -- (5) (2) -- (6) (3) -- (7) (0) -- (5) --(2) -- (7) (4) -- (1) -- (6) -- (3);
\node[left] at (-.7,0) {$H_2:$};
\end{tikzpicture}\quad
\begin{tikzpicture}[inner sep=0pt,baseline=0pt,scale=.9]
\node (0) at (-1,0) {\scriptsize $\bullet$};
\node (1) at (-1.5,.866) {\scriptsize $\bullet$};
\node (2) at (-2.5,.866) {\scriptsize $\bullet$};
\node (3) at (-3,0) {\scriptsize $\bullet$};
\node (4) at (-2.5,.-.866) {\scriptsize $\bullet$};
\node (5) at (-1.5,-.866) {\scriptsize $\bullet$};
\draw (0) -- (1) -- (2) -- (3) -- (4) -- (5) -- 
      (0) -- (2) -- (4) -- (0) (1) -- (3) -- (5) -- (1)
      (0) -- (3) (1) -- (4) (2) -- (5);
\node (x0) at (.5,0) {\scriptsize $\bullet$};
\node (x1) at (1,.866) {\scriptsize $\bullet$};
\node (x5) at (1,-.866) {\scriptsize $\bullet$};
\node (x2) at (2,.866) {\scriptsize $\bullet$};
\node (x3) at (2.5,0) {\scriptsize $\bullet$};
\node (x4) at (2,.-.866) {\scriptsize $\bullet$};
\draw (x0) -- (x1) -- (x2) -- (x3) -- (x4) -- (x5) -- 
      (x0) -- (x2) -- (x4) -- (x0) (x1) -- (x3) -- (x5) -- (x1)
      (x0) -- (x3) (x1) -- (x4) (x2) -- (x5)
      (0) -- (x0) (1) -- (x1) (5) -- (x5) 
      (0) -- (x5) (5) -- (x0)
      (0) -- (x1) (1) -- (x0) 
      (1) -- (x5) (5) -- (x1);
\node[left] at (-3.2,0) {$H_3:$};  
\end{tikzpicture}\\[.5em]
\begin{tikzpicture}[inner sep=1pt,baseline=0pt,scale=1.1]
\node[left=2pt] at (-1.6,0) {$H_4$:};
\node (0) at (-1.5,.3) {\small\tt 0};
\node (1) at (-1.5,-.3) {\small\tt 1};
\node (2) at (-.5,0.4) {\small\tt 2};
\node (3) at (-.5,-.4) {\small\tt 3};
\node (x0) at (1.5,.4) {\small\tt 8};
\node (x1) at (1.5,-.4) {\small\tt 9};
\node (x2) at (2.6,0.3) {\small\tt 10};
\node (x3) at (2.6,-.3) {\small\tt 11};
\node (x12) at (3,-1.4) {\small\tt 12}; 
\node (x13) at (3.7,-1.4) {\small\tt 13}; 
\node (x14) at (3.7,1.4) {\small\tt 14}; 
\node (x15) at (3,1.4) {\small\tt 15}; 
\node (4) at (0,.-1.2) {\small\tt 4};
\node (5) at (1,-1.2) {\small\tt 5};
\node (6) at (1,1.2) {\small\tt 6};
\node (7) at (0,1.2) {\small\tt 7};
\draw (x0) -- (x15) -- (x1) -- (x14) -- (x0) -- (x13) -- (x1) -- (x12) -- (x0);
\draw (x2) -- (x15) -- (x3) -- (x14) -- (x2) -- (x13) -- (x3) -- (x12) -- (x2);
\draw (x12) -- (x13) -- (x14) -- (x15) -- (x12) -- (x14) (x13) -- (x15);
\draw (x0) -- (7) -- (x1) -- (6) -- (x0) -- (4) -- (x1) -- (5) -- (x0) -- (2) -- (x1) -- (3) -- (x0) -- (x1)
(x2) -- (6) -- (x3) -- (x0) -- (x2) -- (x1) -- (x3) -- (5) -- (x2) -- (x3);
\draw (0) -- (1) -- (2) -- (3) -- (0) -- (2) (1) -- (3) -- (4) -- (5) -- (6) -- (7) -- (4) -- (6) (x12) -- (5) -- (7) -- (2) (6) -- (0) -- (7) (5) -- (1) -- (4) -- (2) (6) -- (2) -- (5) -- (3) -- (6) (5) .. controls (-.5,0) .. (0) -- (4) (3) -- (7) -- (1) .. controls (-.5,0) .. (6) -- (x15);
\draw (6) .. controls (2.1,0) .. (x12) (5)  .. controls (2.1,0) .. (x15);
\end{tikzpicture}
\end{center}

\begin{center}
\begin{tabular}{|c|c|r|r|r|} \hline
&&Number of& Number of maximal&\\
Graph&Order &vertices in $S$& clique-partitions&Runtime\\ \hline
$H_4$&16&8 &226&0.018\\
$H_5$&20&10 &962&0.028\\
$H_6$&24&12 &3970&0.057 \\
$H_7$&28&14 &16130&0.122\\
$H_8$&32&16 &65026&0.241\\
$H_9$&36&18 &261122&0.723\\
$H_{10}$&40&20 &1046530&2.280\\
$H_{11}$&44&22 &2094082&3.944\\ \hline
\end{tabular}
\end{center}

 \item Graphs $G_{m,n}$ with set of vertices $V=\{v_1,v_2,\ldots,v_{m+n-1}\}$ and $n$ maximal cliques $\oC_1,\ldots,\oC_n$ such that $\oC_i=\{v_j\mid i\leq j<i+m\}$ for all $1\leq i\leq n$. 
Examples of graphs $G_{m,n}$ are:
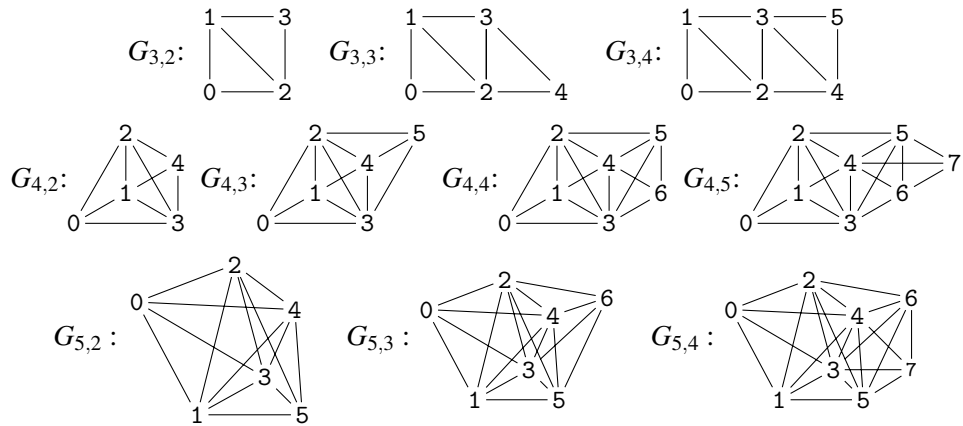
\begin{center}
\begin{tikzpicture}[inner sep=1.5pt,baseline=0pt]
\node (0) at (-1,.5) {\small\tt 1};
\node (1) at (-1,-.5) {\small\tt 0};
\node (2) at (0,-.5) {\small\tt 2};
\node[left=2pt] at (-1.2,0) {$G_{3,2}$:};
\node (3) at (0,.5) {\small\tt 3};
\draw (2) -- (1) -- (0) -- (2) -- (3) -- (0);
\end{tikzpicture}\quad
\begin{tikzpicture}[inner sep=1.5pt,baseline=0pt]
\node (0) at (-1,.5) {\small\tt 1};
\node (1) at (-1,-.5) {\small\tt 0};
\node (2) at (0,-.5) {\small\tt 2};
\node[left=2pt] at (-1.2,0) {$G_{3,3}$:};
\node (3) at (0,.5) {\small\tt 3};
\node (4) at (1,-.5) {\small\tt 4};
\draw (2) -- (1) -- (0) -- (2) -- (3) -- (0) (2) -- (3) -- (4) --(2);
\end{tikzpicture}\quad
\begin{tikzpicture}[inner sep=1.5pt,baseline=0pt]
\node (0) at (-1,.5) {\small\tt 1};
\node (1) at (-1,-.5) {\small\tt 0};
\node (2) at (0,-.5) {\small\tt 2};
\node[left=2pt] at (-1.2,0) {$G_{3,4}$:};
\node (3) at (0,.5) {\small\tt 3};
\node (5) at (1,.5) {\small\tt 5};
\node (4) at (1,-.5) {\small\tt 4};
\draw (2) -- (1) -- (0) -- (2) -- (3) -- (0) (2) -- (3) -- (4) --(2) (3) -- (5) -- (4);
\end{tikzpicture}\\[.5em]
\begin{tikzpicture}[inner sep=1.5pt,baseline=0pt,scale=.8]
\node (0) at (-0.866,-.5) {\small\tt 0};
\node (1) at (0,1) {\small\tt 2};
\node (2) at (0.866,-.5) {\small\tt 3};
\node (4) at (0.866,0.5) {\small\tt 4};
\node (3) at (0,0) {\small\tt 1};
\node[left=2pt] at (-.8,0.2) {$G_{4,2}$:};
\draw (2) -- (1) -- (0) -- (2) -- (3) -- (0) 
     (4) -- (3) -- (1)  (2) -- (4) -- (1);
\end{tikzpicture}\ 
\begin{tikzpicture}[inner sep=1.5pt,baseline=0pt,scale=.8]
\node (0) at (-0.866,-.5) {\small\tt 0};
\node (1) at (0,1) {\small\tt 2};
\node (2) at (0.866,-.5) {\small\tt 3};
\node (4) at (0.866,0.5) {\small\tt 4};
\node (3) at (0,0) {\small\tt 1};
\node (5) at (1.732,1) {\small\tt 5};
\node[left=2pt] at (-.8,0.2) {$G_{4,3}$:};
\draw (2) -- (1) -- (0) -- (2) -- (3) -- (0) 
      (5) -- (4) -- (3) -- (1) -- (5) -- (2) -- (4) -- (1);
\end{tikzpicture}\ 
\begin{tikzpicture}[inner sep=1.5pt,baseline=0pt,scale=.8]
\node (0) at (-0.866,-.5) {\small\tt 0};
\node (1) at (0,1) {\small\tt 2};
\node (2) at (0.866,-.5) {\small\tt 3};
\node (5) at (1.732,1) {\small\tt 5};
\node (6) at (1.732,0) {\small\tt 6};
\node (4) at (0.866,0.5) {\small\tt 4};
\node (3) at (0,0) {\small\tt 1};
\node[left=2pt] at (-.8,0.2) {$G_{4,4}$:};
\draw (2) -- (1) -- (0) -- (2) -- (3) -- (0) 
      (2) -- (6) -- (5) (6) -- (4)
      (5) -- (4) -- (3) -- (1) -- (5) -- (2) -- (4) -- (1);
\end{tikzpicture}\ 
\begin{tikzpicture}[inner sep=1.5pt,baseline=0pt,scale=.8]
\node (0) at (-0.866,-.5) {\small\tt 0};
\node (1) at (0,1) {\small\tt 2};
\node (2) at (0.866,-.5) {\small\tt 3};
\node (5) at (1.732,1) {\small\tt 5};
\node (6) at (1.732,0) {\small\tt 6};
\node (7) at (2.598,.5) {\small\tt 7};
\node (4) at (0.866,0.5) {\small\tt 4};
\node (3) at (0,0) {\small\tt 1};
\node[left=2pt] at (-.8,0.2) {$G_{4,5}$:};
\draw (6) -- (7) -- (5) (7) -- (4)
      (2) -- (1) -- (0) -- (2) -- (3) -- (0) 
      (2) -- (6) -- (5) (6) -- (4)
      (5) -- (4) -- (3) -- (1) -- (5) -- (2) -- (4) -- (1);
\end{tikzpicture}\\[.5em]
\begin{tikzpicture}[inner sep=1.5pt,baseline=0pt],scale=.8
\node[left] at (-1.7,0) {$G_{5,2}:$};
\node (0) at (-1.5,.5) {\small\tt 0};
\node (1) at (-.7,-1) {\small\tt 1};
\node (2) at (-.2,1) {\small\tt 2};
\node (3) at (.2,-.5) {\small\tt 3};
\node (4) at (.6,.4) {\small\tt 4};
\node (5) at (.7,-1) {\small\tt 5};
\draw (0) -- (1) -- (2) -- (3) -- (4) -- (0) -- (2) -- (4) -- (1) -- (3) -- (0) (4) -- (5) -- (1) (3) -- (5) -- (2);
\end{tikzpicture}\quad
\begin{tikzpicture}[inner sep=1.5pt,baseline=0pt,scale=.8]
\node[left] at (-1.7,0) {$G_{5,3}:$};
\node (0) at (-1.5,.5) {\small\tt 0};
\node (1) at (-.7,-1) {\small\tt 1};
\node (2) at (-.2,1) {\small\tt 2};
\node (3) at (.2,-.5) {\small\tt 3};
\node (4) at (.6,.4) {\small\tt 4};
\node (5) at (.7,-1) {\small\tt 5};
\node (6) at (1.5,.7) {\small\tt 6};
 (5) at (.7,-1) {\small\tt 5};
\draw (0) -- (1) -- (2) -- (3) -- (4) -- (0) -- (2) -- (4) -- (1) -- (3) -- (0) (3) -- (6) -- (4) -- (5) -- (1) 
(3) -- (5) -- (2) -- (6) -- (5);
\end{tikzpicture}\quad
\begin{tikzpicture}[inner sep=1.5pt,baseline=0pt,scale=.8]
\node[left] at (-1.7,0) {$G_{5,4}:$};
\node (0) at (-1.5,.5) {\small\tt 0};
\node (1) at (-.7,-1) {\small\tt 1};
\node (2) at (-.2,1) {\small\tt 2};
\node (3) at (.2,-.5) {\small\tt 3};
\node (4) at (.6,.4) {\small\tt 4};
\node (5) at (.7,-1) {\small\tt 5};
\node (6) at (1.5,.7) {\small\tt 6};
\node (7) at (1.5,-.5) {\scriptsize\tt 7};
\draw (0) -- (1) -- (2) -- (3) -- (4) -- (0) -- (2) -- (4) -- (1) -- (3) -- (0) (3) -- (6) -- (4) -- (5) -- (1) 
(3) -- (5) -- (2) -- (6) -- (5) -- (7) -- (6) (3) -- (7) -- (4);
\end{tikzpicture}
\end{center}

\begin{center}
\begin{tabular}{|c|c|r|r|r|} \hline
&&Number of& Number of maximal&\\
Graph&Order &vertices in $S$& clique-partitions&Runtime\\ \hline
$G_{3,2}$&4&2&4&0.001\\
$G_{3,3}$&5&3 &5&0.001 \\
$G_{3,4}$&6&4 &7&0.002 \\
$G_{4,5}$&8&6 &33&0.008 \\
% $G_{4,6}$&9&7 &46&0.013 \\
$G_{5,3}$&7&5 &35&0.003 \\
$G_{5,4}$&8&6 &65&0.006 \\
$G_{5,5}$&9&7 &114&0.015 \\
$G_{5,6}$&10&8 &200&0.031 \\
$G_{6,6}$&11&9 &781&0.037 \\
$G_{7,5}$&11&9 &1488&0.032 \\
$G_{7,6}$&12&10 &3135&0.082 \\
$G_{8,6}$&13&11&12913&0.173\\
$G_{9,6}$&14&12&54495 &0.408  \\ 
$G_{11,9}$&19&17&20899403&121.135 \\
$G_{12,8}$&19&17&40778092&202.608 \\
\hline
\end{tabular}
\end{center}
\end{enumerate}
The Mathematica implementation is slightly slower, but it enables a graphical visualization of the computed results. Some test examples can be seen online at

\centerline{\href{https://staff.fmi.uvt.ro/~mircea.marin/software/MCP.html}{\tt staff.fmi.uvt.ro/\textasciitilde{}mircea.marin/software/MCP.html}}

\section{Conclusion}

We developed an algorithm for computing all maximal clique-partitions in an undirected graph. The algorithm starts from the maximal clique cover of the graph and revises it, reducing the number of vertices shared among cliques by assigning them to one of the cliques they belong to. In this process, we avoid the computation of undesirable answers by detecting and discarding the search states which produce only non-maximal clique-partitions or duplicate answers. Our algorithm is optimal in the following sense: it enumerates all maximal clique-partitions, and each of them is computed only once. The set of computed partitions can be exponentially large with respect to the number of vertices shared among maximal cliques, but every answer can be computed in polynomial time (starting from all maximal cliques). Besides, the computations of different maximal clique-partitions corresponds to the exploration of different branches of the search tree for solutions, and they can be carried out independently, in parallel of each other. % \cite{BHASKER19911} 
Our algorithm is iterative in the sense that it enumerates the computed answers one by one, on demand. This is highly desirable because the total number of maximal clique-partitions can be exponentially large and we may want to start analyzing and processing them as soon as possible.

In many practical applications, we are not interested to enumerate all maximal clique-partitions. Often, our algorithm can be easily adjusted to reduce the search space and compute only the preferred ones. For instance, we can impose the constraint to keep a group of vertices in the same clique by 
assigning them simultaneously (whenever possible) to the same single clique originating from a maximal clique of the graph.
%the  Guiding by heuristics for choosing shared nodes, one can give the priority to one kind of partitions over the others, computing the preferred ones earlier. 

%\nocite{*}
\bibliographystyle{eptcs}
%\bibliography{generic}

\end{document}